\documentclass{SCIS2021}

\usepackage[english]{babel}
\usepackage{amsthm}
\usepackage[colorlinks,linkcolor=black]{hyperref}
\usepackage{graphicx}          
\usepackage{amsmath,amssymb,amsfonts}  
\newtheoremstyle{exampstyle}
{0.0em} 
{0.0em} 
{} 
{1em} 
{\bfseries} 
{.} 
{1em} 
{} 
\usepackage{balance}
\theoremstyle{exampstyle}

\usepackage{CJKutf8}


\begin{document}
	\ArticleType{RESEARCH PAPER}
	\Year{2021}
	\Month{}
	\Vol{}
	\No{}
	\DOI{}
	\ArtNo{}
	\ReceiveDate{}
	\ReviseDate{}
	\AcceptDate{}
	\OnlineDate{}
	
	\title{Robustness Quantification of MIMO-PI Controller From the Perspective of \(\gamma\)-Dissipativity}
	
	\author[1]{Zimao Sheng}{{hpShengZimao@163.com}}
	\AuthorMark{Zimao Sheng}
	
	\AuthorCitation{}
	
	

	\address[1]{Northwestern Polytechnical University,  Xi'an {\rm 710072}, China}
	
	\abstract{
		The proportional-integral-derivative (PID) controller and its variants are widely used in control engineering, but they often rely on linearization around equilibrium points and empirical parameter tuning, making them ineffective for multi-input-multi-output (MIMO) systems with strong coupling, intense external disturbances, and high nonlinearity. Moreover, existing methods rarely explore the intrinsic stabilization mechanism of PID controllers for disturbed nonlinear systems from the perspective of modern robust control theories such as dissipativity and $\mathcal{L}_2$-gain. To address this gap, this study focuses on $\gamma$-dissipativity (partially equivalent to $\mathcal{L}_2$-gain) and investigates the optimal parameter tuning of MIMO-PI controllers for general disturbed nonlinear MIMO systems. First, by integrating dissipativity theory with the Hamilton-Jacobi-Isaacs (HJI) inequality, sufficient conditions for the MIMO-PI-controlled system to achieve $\gamma$-dissipativity are established, and the degree of $\gamma$-dissipativity in a local region containing the origin is quantified. Second, an optimal parameter tuning strategy is proposed, which reformulates the $\gamma$-dissipativity optimization problem into a class of standard eigenvalue problems (EVPs) and further converts it into linear matrix inequality (LMI) formulations for efficient online computation. Comprehensive simulation experiments validate the effectiveness and optimality of the proposed approach. This work provides a theoretical basis for the robust stabilization of general disturbed nonlinear MIMO systems and enriches the parameter tuning methods of PID controllers from the perspective of dissipativity.
		}
	\keywords{MIMO-PI controller, $\gamma$-dissipativity, parameter tuning, disturbed nonlinear MIMO system, eigenvalue problem (EVP)}
	
	\maketitle

	\section{INTRODUCTION}
	
	\subsection{Motivation}
	For a long time, the classical proportional-integral-derivative (PID) controller has played an irreplaceable role in the field of control systems engineering\cite{7823045}. Currently, mainstream variant methods of PID controllers tend to linearize and decouple complex nonlinear systems around their equilibrium points for approximation. Based on equilibrium point characteristics, parameter tuning strategies are designed using experimental methods such as Ziegler-Nichols rules\cite{WOS:000709348200001}, frequency-domain internal model control (IMC)\cite{vilanova2012pid}, and the Indirect Design Approach (IDA)\cite{verma2019robust}. However, these methods struggle to stabilize multi-input-multi-output (MIMO) systems with strong coupling between state variables, intense external disturbances, and high nonlinearity. Moreover, most of these methods are empirical and fail to investigate the intrinsic mechanism by which PID controllers can stabilize disturbed nonlinear systems from the perspectives of passivity and the $\mathcal{L}_2$-gain of system disturbance energy in modern robust control theory\cite{byrnes1991passivity,van19922}. For disturbed nonlinear systems, robustly stable systems should typically exhibit dissipativity\cite{9740606}, meaning that the abstract "energy storage functions" of the system should be attenuated to a certain extent under the influence of disturbance noise. Systems with this property can often achieve a certain degree of Lyapunov stability and Input-to-State Stability (ISS)\cite{SONTAG1995351,9740606} under specific conditions. We have further proven that such dissipative characteristics ensure asymptotic stability against bounded-energy disturbances in the $\mathcal{L}_2$ space. In this study, we start from a special type of $\gamma$-dissipativity, which is partially equivalent to the $\mathcal{L}_2$-gain of the system. The key questions then arise: For general disturbed nonlinear systems, what PID controller parameters should be designed to endow the entire system with dissipativity? Furthermore, how to optimize the controller parameters to maximize the system's dissipativity?
	In our previous work\cite{SHENG2025108152}, we designed a MIMO-PI controller for general disturbed nonlinear systems and obtained the optimal parameter tuning strategy by optimizing the exponential rate at which control system errors converge to the global attractor near the origin and the size of the attractor. However, this method did not involve the optimization of system dissipativity near the origin or even along the entire state trajectory. This inspires us to propose a parameter design method for MIMO-PI controllers with optimal $\gamma$-dissipativity, targeting general disturbed nonlinear MIMO systems from the perspective of $\gamma$-dissipativity.

	\subsection{Related work}
	In contrast to classical PID controllers designed for linear single-input single-output (SISO) plants, significant research efforts have been dedicated in recent years to investigating the robust stabilization problem of PID controllers for nonlinear multi-input multi-output (MIMO) systems subject to disturbances. For instance, Zhong et al. \cite{WOS:000708089400001} proposed an active disturbance rejection control (ADRC)-based tuning rule for MIMO non-affine uncertain systems with model uncertainty, to enable robust stabilization and decoupling control of the systems. Xiang et al. \cite{10989590} incorporated adaptive single-parameter tuning to eliminate the requirement for manual trial-and-error tuning in the presence of nonlinearities and uncertainties. Wang et al. \cite{WOS:001310822800146} further developed a deep reinforcement learning (DRL)-based adaptive PID tuning strategy for setpoint tracking in systems with time-varying uncertainty. However, these investigations primarily target nonlinear systems with specific structures, and most primarily address model uncertainty as opposed to externally imposed disturbances. For systems with external disturbances and nonlinearities, existing research \cite{WOS:001240938200001, WOS:001355007600001, WOS:001485121000042} tends to directly linearize the nonlinear systems into linear time-invariant (LTI) models for controller parameter design. This approach aggregates nonlinear terms and external disturbances into a single component, treating them collectively as noise. Nevertheless, the controller parameters obtained via linearization around the equilibrium point often struggle to cope with plants exhibiting significant nonlinear characteristics.
	
	Theoretically, it is fully achievable to accomplish local or even global robust stabilization of nonlinear disturbed MIMO systems through appropriate controller parameter tuning. For example, Zhao et al. \cite{bib:Zhao, bib:Zhao2, zhao2017pid} theoretically established sufficient conditions for guaranteeing effective error stabilization of PID controllers for general nonlinear uncertain MIMO systems. Jinke et al. \cite{bib:Zhang} further proposed an explicit method to construct PID parameters by leveraging the upper bounds of derivatives of unknown nonlinear drift and diffusion terms in stochastic mechanical systems with white noise and linear inputs.
	
	However, the core challenge lies not only in ensuring asymptotic stability or even exponential stabilization of the error dynamics for disturbed systems. Instead, we aim to design a suite of practically viable metrics that can quantify the robustness performance of MIMO-PI controllers when stabilizing the most general nonlinear disturbed systems. On this basis, we can optimize the system robustness by adjusting the controller parameters. It is worth noting that in our previous work \cite{SHENG2025108152}, we have designed optimal robust controllers adopting two robustness metrics: the exponential convergence rate of the disturbed system error and the size of the ultimate invariant set of convergence.
	
	Classical robust controller design methods typically aim to minimize the \(H_\infty\) norm of the disturbed system to mitigate the impact of disturbances. For example, Gopmandal et al. \cite{WOS:000654370700001} enhanced the hybrid search approach based on LTI static output feedback (SOF) from the perspective of \(H_\infty\)-synthesis for MIMO linear time-variant (LTV) systems with norm-bounded time-varying uncertain matrices. Nevertheless, no universally accepted metric exists for quantifying the \(H_\infty\) norm of nonlinear disturbed systems, which has posed a bottleneck for the optimal parameter tuning of MIMO-PI controllers for such systems.

	\subsection{Contributions}
	Analogous to the \(H_\infty\) norm for linear systems, concepts such as \(\gamma\)-dissipativity \cite{moylan2014dissipative, brogliato2007dissipative} and \(\mathcal{L}_2\)-gain \cite{van2000l2} in dissipativity theory enable the characterization of energy dissipation behaviors of disturbed nonlinear dynamical systems in an abstract framework, with more relevant results available in Ref. \cite{9740606}. Building upon our prior work, we present the following contributions from the perspective of \(\gamma\)-dissipativity of disturbed nonlinear dynamical systems:
	\begin{itemize}
		\item For the most general nonlinear MIMO model with external disturbances, we establish sufficient conditions for the MIMO-PI-controlled system to attain \(\gamma\)-dissipativity through the integration of dissipativity theory and the Hamilton-Jacobi-Isaacs (HJI) inequality, and quantify the degree of \(\gamma\)-dissipativity in a local region containing the origin;
		\item From the perspective of optimizing the \(\gamma\)-dissipativity of the controlled system, we develop an optimal parameter tuning strategy for the MIMO-PI controller. This strategy reformulates the \(\gamma\)-dissipativity optimization problem of the disturbed nonlinear MIMO system as a class of standard eigenvalue problems (EVPs), which can be further converted into linear matrix inequality (LMI) formulations for efficient online computation. Comprehensive simulation experiments validate the effectiveness and optimality of the proposed approach. 
	\end{itemize}
	
	\section{PRELIMINARIES}
	In this section, we introduce the concept of \( \gamma \)-dissipativity, which is a key component of modern robust control theory, and briefly describe its essential property. Building on this, we propose a MIMO-PI controller for more generalized multi-input multi-output (MIMO) nonlinear perturbed systems, which serve as the foundation for subsequent developments.
	\subsection{\( \gamma \)-dissipativity}
	In this section, we first introduce the concepts of a special class of \( \gamma \)-dissipativity and \( \mathcal{L}_2 \)-gain, as presented in Definition \ref{pr:prop1} and \ref{pr:prop2}. These concepts are natural extensions of the \( H_{\infty} \) norm for linear perturbed systems in modern control theory. Furthermore, we establish the connection between them as shown in Lemma \ref{lem:lem1}. In Lemma \ref{lem:lem3}, we prove that nonlinear \( \gamma \)-dissipativity systems are asymptotically stable with respect to finite-energy disturbances under certain conditions. Finally, we present a sufficient condition for a nonlinear perturbed system to be \( \gamma \)-dissipativity, known as the HJI inequality (see Lemma \ref{lem:HJI}).
	\newtheorem{prop1}{Definition}
	\begin{prop1}
		\label{pr:prop1}
		(\( \gamma \)-dissipativity) For the nonlinear error system with disturbances $\dot{e}=f(e,\Omega)$, if there exists a smooth, continuously differentiable, and positive semi-definite storage function \( V(e)\geq 0 \) satisfying \( V(0) = 0 \), such that the inequality
		\begin{align}
			\label{eq:gamma_dissipative}
			\dot{V}(e) \leq \frac{1}{2}\left( \gamma^2 \|\Omega\|^2 - \|e\|^2 \right)
		\end{align}
		holds for all admissible disturbance inputs \( \Omega \) (where \( \gamma > 0 \)), then the system is referred to as \( \gamma \)-dissipativity.
	\end{prop1}
	\newtheorem{prop2}[prop1]{Definition}
	\begin{prop2}
		\label{pr:prop2}
		(\( \mathcal{L}_2 \)-gain) For the nonlinear perturbed system $\dot{e}(t) = f(e(t),\Omega(t))$, $t\in [0,T]$, state $e(t)\in \mathbb{R}^n$, $e(t_0)=e_0$, disturbance $\Omega(t) \in \mathbb{R}^m$, if there exists $\gamma >0$ such that
		\begin{align}
			\sqrt{\frac{\int_{0}^{T}||e(\tau)||^2d\tau  }{\int_{0}^{T}||\Omega(\tau)||^2d\tau } }\leq \gamma
		\end{align}
		then the system is said to exhibits the \( \mathcal{L}_2 \)-gain property.
	\end{prop2}
	
	\newtheorem{lem1}{Lemma}
	\begin{lem1}
		\label{lem:lem1} 
		For a given \( \gamma > 0 \), if the disturbed system is \( \gamma \)-dissipativity for zero-initial state $e(0)=0$, then it possesses an \( \mathcal{L}_2 \)-gain less than or equal to \( \gamma \).  
	\end{lem1}
	\begin{proof}
		See Appendix \ref{app:pf_Lemma_1} for the proof.
	\end{proof}
	\newtheorem{lem2}[lem1]{Lemma}
	\begin{lem2}
		\label{lem:lem3}
		(Asymptotic stability of \( \gamma \)-dissipativity system) For the nonlinear system $\dot{e}(t) = f(e(t),\Omega(t))$, $t\in \mathbb{R}^+$, $e(t)\in \mathbb{R}^n$, $e(t_0)=e_0$, $\Omega(t) \in \mathbb{R}^m$, if the following conditions are satisfied: 
		(1) there exists a smooth, continuously, and positive semi-definite storage function \( V(e(t))\geq 0 \) satisfying \( V(0) = 0 \), such that the system is \( \gamma \)-dissipativity;
		(2) there exists $\bar{\lambda} >0$ such that $ V(e(t)) \leq \bar{\lambda} ||e(t)||^2$;
		(3) $\int_0^{\infty}||\Omega(\tau)||^2 d\tau = C_{\Omega}<\infty$.
		Then, $e(t) \rightarrow \left\{
		e:V(e) = 0
		\right\}$, as $t\rightarrow \infty$. Besides, if there exists $\alpha \in \mathcal{K}(.)$ such that $V(e(t)) \geq \alpha\left(||e(t)|| \right)$,then 
		 $e(t) \rightarrow 0$ as $t\rightarrow \infty$.
	\end{lem2}
	\begin{proof}
		See Appendix \ref{app:pf_Lemma_3} for the proof.
	\end{proof}
	
	It is evident that systems with the \( \gamma \)-dissipativity property exhibit strong robustness, which is reflected not only in ensuring a bounded \( \mathcal{L}_2 \)-gain for the disturbed system but also in guaranteeing its asymptotic stability under finite-energy disturbances. Additionally, this robustness can be quantified by the \( \gamma \). Next, a sufficient condition for determining whether the disturbed system is \( \gamma \)-dissipativity is provided.
	
	\newtheorem{lem3}[lem1]{Lemma}
	\begin{lem3}
		\label{lem:HJI}
		(HJI-inequality) The necessary and sufficient condition for system $\dot{e} = f(e) + g(e)\omega$, $z = h(e)$,
		where the state $e\in \mathbb{R}^n$, $f(0) = 0$, the disturbance $\omega\in \mathbb{R}^m$, to be \( \gamma \)-dissipativity is that there exists a smooth and differentiable storage function \( V(e) \) such that the following Hamilton-Jacobi-Issacs(HJI)-inequality 
		\begin{align}
			\frac{\partial V}{\partial e}f(e) + \frac{1}{2\gamma^2}\frac{\partial V}{\partial e}g(e)g^T(e)\left[
			\frac{\partial V}{\partial e}
			\right]^T + \frac{1}{2}h^T(e)h(e) \leq 0
		\end{align}
		admits a semi-positive definite \( V(e) \geq 0 \) with \( V(0) = 0 \).  
	\end{lem3}
	\begin{proof}
		See Ref. \cite{9740606} for the proof.
	\end{proof}
	
	\newtheorem{schur}[lem1]{Lemma}
	\begin{schur}
		\label{thm:schur}
		(Schur Complement Lemma) For a given symmetric matrix $S=\begin{bmatrix}
			S_{11} & S_{12} \\ S_{21} & S_{22}
		\end{bmatrix}$, where $S_{11} \in \mathbb{R}^{r\times r}$. The following three conditions are equivalent. (1) $S<0$; (2) $S_{11}<0,S_{22} - S_{12}^TS_{11}^{-1}S_{12}<0$; (3) $S_{22}<0,S_{11} - S_{12}S_{22}^{-1}S_{12}^T<0$. 
	\end{schur}
	\begin{proof}
		See Ref. \cite{fuzhen2005schur} for the proof.
	\end{proof}
	
	\subsection{Eigenvalues of Hurwitz matrix}
	In the following Lemma \ref{lem:lemma_appendix2}, we present a method to determine the existence of a common characteristic matrix that satisfies the Lyapunov inequality for all countably many Hurwitz matrices. Notably, this lemma plays a crucial role in verifying the existence of the common matrix $P$ involved in Theorem \ref{lem:lem4}.
	\newtheorem{lem4}[lem1]{Lemma}
	\begin{lem4}
		\label{lem:lemma_appendix2}
		For $A_i\in \mathbb{R}^{n\times n}$, $i=1,2,\ldots N$, $L_{A,i} = \sup_{1\leq j\leq N}\left\| A_j - A_i\right\|$, if $Re\left[\lambda\left( A_i\right)\right]\leq -\varepsilon$, and
		$\varepsilon > \min_{1\leq i\leq N}\left\{
		L_{A,i}M^2\left(A_i\right)
		\right\}$ where $\left\|e^{A_i^T t}  \right\|\leq M\left(A_i\right)e^{-\varepsilon t}$, then there exists $P=P^T>0$ such that $A_i P + P A_i^T + \varepsilon^* I \leq O$ for any $A_i,\varepsilon^*>0$ holds. 
	\end{lem4}
\begin{proof}
	\label{app:pf_Lemma_6}
	We can construct such a 
	\begin{align}
		P=\int_0^{\infty}e^{A_0^T t}Qe^{A_0 t}dt	
	\end{align}
	where
	\begin{align}
		A_0 = \text{argmin}_{i=1,2\ldots N}\left\{
		\left[\sup_{j}\left\|A_j-A_i^T\right\| \right]M^2\left(A_i\right)
		\right\}
	\end{align}
	It's obvious that $P=P^T>O$, let $F(t)=e^{A_0^T t}Qe^{A_0 t}$, $Q=Q^T>O$, its differential can be described as
	\begin{align}
		\frac{d F(t)}{dt} = A_0^Te^{A_0^T t}Qe^{A_0 t} + e^{A_0^T t}Qe^{A_0 t}A_0
	\end{align}
	Meanwhile, define $S_i(t) = A_ie^{A_0^T t}Qe^{A_0 t} + e^{A_0^T t}Qe^{A_0 t}A_i^T$, and
	\begin{align}
			T_i(t) &= S_i(t) - \frac{dF(t)}{dt} \\
			&=\underbrace{\left(A_i-A_0^T \right)}_{K_i}e^{A_0^T t}Qe^{A_0 t} + e^{A_0^T t}Qe^{A_0 t}\underbrace{\left(A_i^T-A_0\right)}_{K_i^T}
	\end{align}
	\begin{align}
		\begin{split}
			A_i P + PA_i^T  
			&= \int_0^{\infty} A_iF(t) + F(t)A_i^Tdt\\
			&= \int_0^{\infty} S_i(t)dt = \int_0^{\infty} T_i(t) + \frac{dF(t)}{dt} dt\\
			&=F(\infty) - F(0) + \int_{0}^{\infty} T_i(t)dt\\
			&=-Q + \int_{0}^{\infty} T_i(t)dt
		\end{split}
	\end{align}
	Moreover, $\forall i$, $\left\|K_i\right\| =\left\|A_i-A_0^T\right\| \leq \sup_{i}\left\|A_i-A_0^T\right\|$, and $\varepsilon > \sup_i \left\|A_i-A_0^T \right\| M^2(A_0) \geq \left\|K_i \right\| M^2(A_0)$, hence,
	\begin{align}
		\begin{split}
			\int_{0}^{\infty} \left\|T_i(t)\right\|dt 
			&\leq \int_0^{\infty} 2\left\| K_i\right\| \left\|e^{A_0^T t}Qe^{A_0 t}\right\| dt\\
			&\leq 2\left\|K_i\right\|M^2\left(A_0\right)\left\| Q\right\|\int_0^{\infty} e^{-2\varepsilon t}dt \\
			& = \frac{\left\|K_i\right\|M^2\left(A_0\right)\left\| Q\right\|}{\varepsilon} < \left\| Q\right\|
		\end{split}
	\end{align}
	Due to the fact that $\varepsilon > \left\|K_i\right\|M^2\left(A_0\right)$, thereby for any $\varepsilon^*>0$, there exists $\alpha >0$ such that
	\begin{align}
		\frac{\left\|K_i\right\|M^2\left(A_0\right)}{\varepsilon} < 1-\frac{\varepsilon^*}{\alpha}
	\end{align}
	For any $x\in\mathbb{R}^n$, let $Q=\alpha I$,
	\begin{align}
		\begin{split}
			x^T\left(A_i P + PA_i^T \right)x & = x^T 
			\left(
			-\alpha I  + \int_0^{\infty}T_i(t) dt  
			\right)x \\
			& \leq x^T\left(-\alpha + \int_0^{\infty} \left\|T_i(t) \right\|   dt \right)x \\
			&\leq x^T\left(-\alpha + \left(1-\frac{\varepsilon^*}{\alpha}\right)\left\| Q\right\|   dt \right)x\\
			&\leq -\varepsilon^* x^Tx
		\end{split} 
	\end{align}
	Hence, $A_i P + PA_i^T + \varepsilon^* I \leq O$ holds for any $A_i$ holds.
\end{proof}
	\newtheorem{cor}{Remark}
	\begin{cor}
		For Lemma \ref{lem:lemma_appendix2}, the inequality $\left\|e^{A^T t}\right\| \leq M(A)e^{-\varepsilon t}$ holds, where $\varepsilon = -\lambda_{\max}\big[Re(A)\big]>0$, $M(A) = \sup_t\left\|e^{(A^T+\varepsilon I)t} \right\|$. By performing the Jordan decomposition on $A^T$, we have $A^T=\mathcal{P}\mathcal{J}\mathcal{P}^{-1}$, where $\mathcal{P}$ is an invertible matrix and $\mathcal{J}$ is a Jordan matrix. Then $M(A)$ can be computed as follows:
		 \begin{align}
		 	\begin{split}
		 		M(A) &= \sup_{t>0}\left\|\mathcal{P} e^{(\mathcal{J}+\varepsilon I)t}\mathcal{P}^{-1} \right\|
		 		 \leq 
		 		 \underbrace{\left\| \mathcal{P}\right\| 
		 		 	\left\| \mathcal{P}^{-1}\right\|
		 		 	\sup_{t>0}\left\|e^{(\mathcal{J}+\varepsilon I)t}  \right\| }_{\tilde{M}(A)}
		 		 ,\quad \mathcal{J}+\varepsilon I\leq O 
		 	\end{split}
		 	\end{align}
		 	We can adopt \( \tilde{M}(A) \) as an estimation of \( M(A) \). Particularly, if the Jordan canonical form \( \mathcal{J} \) of \( A^T \) is a purely diagonal matrix (i.e., \( A^T \) is diagonalizable), then we have
		 	$
		 	\sup_{t>0}\left\|e^{(\mathcal{J}+\varepsilon I)t}  \right\|=1.
		 	$
		 	In this case, the constant \( M(A) \) in the inequality simplifies to the product of the norm of the similarity transformation matrix and its inverse:
		 	$
		 	M(A)=\left\| \mathcal{P}\right\| \cdot \left\| \mathcal{P}^{-1}\right\|
		 	$.
	\end{cor}
	
	\subsection{MIMO-PI Controller}
	In this section, we analyze a general perturbed nonlinear system and propose a MIMO-PI controller\cite{SHENG2025108152}. Consider the following perturbed autonomous MIMO general nonlinear system within continuous and first-order differentiable $f \in \mathbb{R}^n$, and first-order differentiable disturbance $\omega(t) \in \mathbb{R}^l$,
	\begin{align}
		\label{eq:perturbed non-affine nonlinear system}
		\dot{e}(t) = f(e(t),u(t)) + \Gamma \omega(t), \ f(0,0)=0
	\end{align}
	where state $e(t)\in \mathbb{R}^n$, $e(t_0)=e_0$, control input $u(t) \in \mathbb{R}^{m}$, and $\Gamma\in\mathbb{R}^{n\times l}$. Further, we assume the adoption of a MIMO-PI controller that takes into account the coupling of multiple input channels. Explicitly, 
	\begin{equation}
		\label{eq:controller}
		u(t)=u(e(t)) = K_P e(t) + K_I \int_0^t e(t) dt
	\end{equation}
	where $K_P \in \mathbb{R}^{m \times n}, K_I \in \mathbb{R}^{m \times n}$. Here, the control commands and its first-order derivative quantity are constrained as
	\begin{equation}
		\label{equ:commands_constrain}
		u_{\min}\leq u(t) \leq u_{\max},  \dot{u}_{\min} \leq \dot{u}(t) \leq \dot{u}_{\max}
	\end{equation}
	where $u_{\min},u_{\max},\dot{u}_{\min},\dot{u}_{\max},\in \mathbb{R}^{m}$. We expect to exponentially stablize the Eq.(\ref{eq:perturbed non-affine nonlinear system}) within minimum \( \gamma \)-dissipativity to resist the sudden disturbances. 
	
	\section{MAIN RESULTS}
	
	\subsection{The \( \gamma \)-dissipativity of MIMO-PI Controller}
	As the core of this paper, Theorem \ref{lem:lem4} not only reveals the rationale behind the \( \gamma \)-dissipativity of general nonlinear perturbed systems under the MIMO-PI controller, but also specifically presents a sufficient condition for the error trajectory to possess \( \gamma \)-dissipativity over a certain time interval. Meanwhile, in Corollary \ref{cor:cor9}, we further provide a sufficient condition for the existence of the characteristic matrix in Theorem \ref{lem:lem4}, more precisely from the perspective of the eigenvalues of the Jacobian of the error trajectory. In particular, this result indicates that the entire trajectory is more likely to exhibit higher \( \gamma \)-dissipativity when the Jacobian matrix has negative eigenvalues that are further from the imaginary axis.
	\newtheorem{thm1}{Theorem}
	\begin{thm1}
		\label{lem:lem4} 
		For the perturbed autonomous system Eq.(\ref{eq:perturbed non-affine nonlinear system}), the MIMO-PI controller Eq.(\ref{eq:controller}) is adopted as the input. If there exists $P=P^T>O$ and $\varepsilon > 1$ for any trajectory $e=e(t)$, $t\geq t_0$ such that
		\begin{align}
			\label{eq:MIMO_PI_condition}
			PA_K(e) + A_K^{T}(e)P + \varepsilon I \leq O
		\end{align}
		Then the system exhibits the \( \gamma^* \)-dissipativity property, here:
		\begin{align}
			\gamma^* = \inf_{\gamma}\left\{
			\gamma: (1-\varepsilon) I  + \frac{1}{\gamma^2}PGG^TP\leq O, \forall e
			\right\}
		\end{align}
		where $K=\left[K_P,K_I\right]$, $A_K(e) = D_1(e) + D_2(e)K$,
		\begin{align}
			D_1(e) = \begin{bmatrix}
				\frac{\partial f}{\partial e} & O\\
				I & O
			\end{bmatrix},
			D_2(e) = \begin{bmatrix}
				\frac{\partial f}{\partial u} \\ O
			\end{bmatrix},
			G = \begin{bmatrix}
				\Gamma \\ O
			\end{bmatrix}
		\end{align}
		Moreover, for finite-energy disturbance $\omega$ that satisfy $\int_{0}^{\infty} \left\|\omega(\tau) \right\|^2 d\tau < \infty$, the error $e(t),\dot{e}(t) \rightarrow 0$ as $t\rightarrow \infty$.
	\end{thm1}
	\begin{proof}
		Substituting the differential form of the MIMO-PI controller $\dot{u}(t) = K_P \dot{e}(t) + K_I e(t)$ into the velocity form of perturbed autonomous system 
		\begin{align}
			\ddot{e}(t) = \frac{\partial f}{\partial e} \dot{e}(t) + \frac{\partial f}{\partial u} \dot{u}(t) + \Gamma \dot{\omega}(t) 
		\end{align}
		to obtain the state-space model of $s(t) = \left[\dot{x}(t)^T,x(t)^T\right]^T$ as
		\begin{align}
			\underbrace{
				\begin{bmatrix}
					\ddot{e}(t) \\ \dot{e}(t)
			\end{bmatrix}}_{\dot{s}(t)}
			=
			\underbrace{
				\begin{bmatrix}
					\frac{\partial f}{\partial e} + \frac{\partial f}{\partial u} K_P & \frac{\partial f}{\partial u} K_I\\ I & O
				\end{bmatrix}
			}_{A_K(e)} 
			\underbrace{
				\begin{bmatrix}
					\dot{e}(t) \\ e(t)
				\end{bmatrix}
			}_{s(t)}
			+
			\underbrace{\begin{bmatrix}
					\Gamma \\ O
			\end{bmatrix}}_{G}\dot{\omega}(t)
		\end{align} 
		here $A_K(e)$ can be decomposed as,
		\begin{align}
			A_K(e) = \underbrace{\begin{bmatrix}
					\frac{\partial f}{\partial e} & O \\
					I & O
			\end{bmatrix}}_{D_1(e)} + \underbrace{
				\begin{bmatrix}
					\frac{\partial f}{\partial u} \\
					O
				\end{bmatrix}
			}_{D_2(e)}\underbrace{\begin{bmatrix}
					K_P & K_I
			\end{bmatrix}}_{K}
		\end{align}
		Due to the fact that, when $\varepsilon >1$:
		\begin{align}
			\label{eq:gamma_star}
			\lim_{\gamma \rightarrow \infty}\left\{
			(1-\varepsilon) I  + \frac{1}{\gamma^2}PGG^TP
			\right\}= (1-\varepsilon) I < O
		\end{align}
		Hence, there exists $\gamma^*$ as shown in Eq.(\ref{eq:gamma_star}) such that
		\begin{align}
			(1-\varepsilon) I  + \frac{1}{\gamma^{*2}}PGG^TP\leq O
		\end{align}
		We construct the positive define Lyapunov function $V(s)=\frac{1}{2}s^TPs$, and its derivative can be expressed as:
		\begin{align}
			\frac{\partial V}{\partial s} = s^TP,\ V(0) = 0
		\end{align}
		Hence,
		\begin{align}
			\begin{split}
				\frac{\partial V}{\partial s}A_K(e)s + \frac{1}{2\gamma^{*2}}\frac{\partial V}{\partial s}GG^T\left[\frac{\partial V}{\partial s} \right]^T &+ \frac{1}{2}s^Ts 
				 = \frac{1}{2}s^T
				\left[
				PA_K(e) + A_K(e)^TP + \frac{1}{\gamma^{*2}}PGG^TP + I
				\right]s	\\
				& = \frac{1}{2}s^T
				\bigg[
				PA_K(e) + A_K^T(e)P + \varepsilon I\bigg]s +\frac{1}{2}s^T
				\bigg[ (1-\varepsilon)I 
				 + \frac{1}{\gamma^{*2}}PGG^TP
				\bigg]s\\
				&\leq \frac{1}{2}s^T
				\bigg[
				PA_K(e) + A_K^T(e)P + \varepsilon I\bigg]s \leq 0
			\end{split}
		\end{align}
		According to Lemma \ref{lem:HJI}, HJI-inequality holds, the system exhibits $\gamma^{*}$-dissipativity. Moreover, according to the Lemma \ref{lem:lem3}, the proof is completed.
	\end{proof}
	
	\newtheorem{cor1}[cor]{Remark}
	\begin{cor1}
		\label{cor:cor9}
		In Theorem \ref{lem:lem4}, to guarantee the existence of a common positive definite \( P \) specified in Eq.(\ref{eq:MIMO_PI_condition}) for all \( e= e(t)\in \Omega \) with \( t \geq t_0 \), we start from Lemma \ref{lem:lemma_appendix2} and derive a sufficient condition for the obtaining of $K$ from the eigenvalue perspective as follows:
		\begin{align}
			\ Re\left[\lambda\left(A_K(e) \right)\right]\leq -\varepsilon,\quad \forall e \in \Omega 
		\end{align}
		where,
		\begin{align}
			\varepsilon > \inf_{e\in\Omega}\left\{
			\sup_{\tilde{e}\in\Omega}\left\| A_K(\tilde{e}) - A_K(e)\right\|
			M^2_K\left( e \right)
			\right\}
		\end{align}
		and $\left\| \exp\left(A_K(e)t \right)\right\|\leq M_K\left(e \right)\exp\left( - \varepsilon t\right)$. Furthermore, since $e(t)\to0$ as $t\to\infty$, it is evident that $0\in\Omega$. Thus, let
		\begin{align}
			S_K(\Omega) =
			\sup_{\tilde{e}\in \Omega}\left\|A_K(\tilde{e}) -A_K(0)\right\|
			 M^2_K\left(0\right)\geq  \inf_{e\in\Omega}\left\{
			 \sup_{\tilde{e}\in\Omega}\left\| A_K(\tilde{e}) - A_K(e)\right\|
			 M^2_K\left( e \right)
			 \right\}
		\end{align}
		we can determine the $\gamma_K^*(\Omega)$-dissipativity of the nonlinear perturbed system by examining the negativity of the index $\mathcal{L}_K(\Omega)$:
		\begin{align}
			\mathcal{L}_K(\Omega) = \sup_{e\in\Omega}\lambda_{\max}\left[Re\left( A_K(e) \right)\right] + S_K(\Omega)< 0
		\end{align}
	\end{cor1}
	Under the above constraint, the dissipativity $\gamma^*_K(\Omega)$ for parameter $K$ can be described as:
	\begin{align}
		\label{eq:gamma_star_K}
		\gamma^*_K(\Omega) = \sup_{\varepsilon_K>1}\inf_{\gamma}\left\{
		\gamma: (1-\varepsilon_K) I  + \frac{1}{\gamma^2}P_K(\varepsilon_K)GG^TP_K(\varepsilon_K)\leq O
		\right\}
	\end{align}
	According to the Proof of Lemma \ref{lem:lemma_appendix2} in Appendix \ref{app:pf_Lemma_6}, $P_K(\varepsilon_K)$ takes the following form: 
	\begin{align}
		\label{eq:P_K_varepsilon_K}
		P_K(\varepsilon_K) = \varepsilon_K\left(1-\frac{S_K(\Omega)}{\mathcal{L}_K(\Omega)}\right) \int_0^{\infty}e^{A_K(0)^T t}e^{A_K(0) t} dt
	\end{align}
	Combining Eq.(\ref{eq:gamma_star_K}) and Eq.(\ref{eq:P_K_varepsilon_K}), we can further derive the specific form of $\gamma_K^*(\Omega)$ as follows:
	\begin{align}
		\label{eq:gamma_K_final}
		\gamma_K^*(\Omega) = \sup_{\varepsilon_K>1}\inf_{\gamma}\left\{
		\gamma:\frac{1-\varepsilon_K}{\varepsilon_K^2\left( 1-\frac{S_K(\Omega)}{\mathcal{L}_K(\Omega)}\right)^2} + \frac{1}{\gamma^2}\tilde{P}_KGG^T\tilde{P}_K\leq O
		\right\}
	\end{align}	
	where $\tilde{P}_K=\int_0^{\infty}e^{A_K(0)^T t}e^{A_K(0) t} dt$, i.e., the solution of Lyapunov function
	 $A_K(0)^T\tilde{P}_K + \tilde{P}_KA_K(0) + I=O$.
	 
	 The above results demonstrate that the MIMO-PI controller with \(K\) as the adjustable parameter exhibits the property of \(\gamma_K^*(\Omega)\)-dissipativity in region \(\Omega\) (encompassing the neighborhood of the origin) when the constraint \(\mathcal{L}_K(\Omega)<0\) is satisfied. Under such a constraint, the dissipativity level of the system can be quantified by \(\gamma_K^*(\Omega)\).
	 Besides, the above conclusions establish a criterion for determining whether the MIMO-PI controller is \(\gamma\)-dissipativity over the given region \(\Omega\). Furthermore, over this \(\gamma\)-dissipativity region, we quantify the degree of dissipativity.
	
	\subsection{Parameter Tuning within Optimal $\gamma$-dissipativity}
	In this section, provided that the MIMO-PI controller renders the nonlinear perturbed system $\gamma$-dissipativity, we investigate the optimization problem of the parameter \(K^*\) from the viewpoint of disturbance energy. Specifically, we aim to design the controller such that the \(\mathcal{L}_2\)-gain over the region \(\Omega\) is minimized as:
	\begin{align}
		K^*=\text{arg}\min_{K}\gamma_K^*(\Omega),\quad s.t.\ \mathcal{L}_K(\Omega) <0,\ A_K(0)^T\tilde{P}_K + \tilde{P}_KA_K(0) + I=O
	\end{align}
	The core difficulty in solving the aforementioned optimization problem lies in determining the \(\mathcal{L}_2\)-gain \(\gamma_K^*(\Omega)\) over the fixed domain \(\Omega\). Furthermore, we can transform the aforementioned \(\gamma_K^*(\Omega)\) into a class of canonical Eigenvalue Value Problem (EVP). Specifically, we aim to select an appropriate gain \(K\) such that there exists a scalar \(\varepsilon_K>1\), so as to attain a minimized \(\gamma\). This optimization problem can be rewrited as follows:
	\begin{align}
		\min_K \gamma, \quad s.t.\ \left(1-\frac{S_K(\Omega)}{\mathcal{L}_K(\Omega)}\right)^2\tilde{P}_K G G^T \tilde{P}_K - \frac{(\varepsilon_K-1)\gamma^2}{\varepsilon_K^2} I \leq O,\ \mathcal{L}_K(\Omega)<0
	\end{align}
	Given that
	\begin{align}
		\max_{\varepsilon_K>1}\frac{(\varepsilon_K-1)}{\varepsilon_K^2}= \frac{1}{4}
	\end{align}
	 we further transform the aforementioned problem into the following Linear Matrix Inequality (LMI) form by virtue of the Schur Complement Lemma:
	 \begin{align}
	 		\label{eq:EVPs}
	 		\gamma_K^*(\Omega)=\min_{K,\mathcal{L}_K(\Omega)<0} \gamma>0,\quad
	 		s.t.\ \begin{bmatrix}
	 		O & 2\left(1-\frac{S_K(\Omega)}{\mathcal{L}_K(\Omega)}\right)\tilde{P}_K G \\
	 		2\left(1-\frac{S_K(\Omega)}{\mathcal{L}_K(\Omega)}\right) G^T \tilde{P}_K & O 
	 	\end{bmatrix}\leq \gamma I
	 \end{align}
	 where $\tilde{P}_K$ satisfys that $K^TD_2(0)^T\tilde{P}_K + \tilde{P}_KD_2(0)K + D_1(0)^T\tilde{P}_K + \tilde{P}_KD_1(0) +I = O$.
	 We aim to realize optimal \(\gamma\)-dissipativity-based parameter tuning over the specified domain \(\Omega\). This objective is attainable if there exists a parameter \(K\) within \(\Omega\) satisfying \(\mathcal{L}_K(\Omega) < 0\). Once this condition is met, one can conduct an exhaustive search over all feasible parameters \(K\) within the domain to retrieve the optimal parameter \(K^* = \arg\max_{K}\gamma_K^*(\Omega)\) of MIMO-PI controller.
	 
\section{SIMULATION}
In this section, we will verify the $\gamma$-dissipativity theory of the designed MIMO-PI controller through a specific control example of a perturbed nonlinear MIMO model. Furthermore, based on this theory, we will carry out the optimal design of controller parameters to verify the rationality of our optimal parameter tuning strategy.
\subsection{Experimental configuration}
 For simplicity, we adopt the perturbed nonlinear controlled model from our prior work \cite{SHENG2025108152} as the benchmark. This model is a simplified kinematic model for designing robust path-following guidance laws of fixed-wing UAVs \cite{bib:Samir} in the ground coordinate system along the $\gamma$ and $\chi$ directions, expressed as follows:
 \begin{align}
 	\label{eq:aircraft guidance law}
 	\begin{split}
 		&\dot{\chi}(t) = \frac{g\tan \phi(t)}{V} + d_{\chi} \\
 		&\dot{\gamma}(t) = \frac{g[n_z(t)\cos \phi(t) - \cos \gamma (t)]}{V} + d_{\gamma}\\
 	\end{split}
 \end{align}
 Herein, the flight states of the kinematic model, $x(t)=\left[\gamma(t),\chi(t) \right]^T$, represent the flight path angle and course angle of the UAV at time $t$, respectively. The control input is defined as $u(t)=[\phi(t),n_z(t)]^T$, a vector comprising the roll angle $\phi(t)$ and the normal overload $n_z(t)$ along the z-axis direction. The term $d = [d_{\chi}, d_{\gamma}]^T$ denotes the disturbance imposed on $\dot{\chi}(t)$ and $\dot{\gamma}(t)$, induced by factors such as wind fields and model simplification. Specifically, the perturbation $d$ takes the form of sinusoidal noise signals: $ d_{\chi} = L_{d_{\chi}}\sin(\omega_{\chi}t) $ and $ d_{\gamma} = L_{d_{\gamma}}\cos(\omega_{\gamma}t) $.
 The aircraft velocity $V$ is assumed to be constant, and $g$ denotes the gravitational acceleration with a value of $\mathrm{9.81m/s^2}$.
 We define the reference tracking signals as $x_c = \left[\chi_c,\gamma_c\right]^T$, with the tracking error $e(t) = \left[e_{\chi}(t),e_{\gamma}(t)\right]^T$ given by $ e(t) = x_c - x(t)= \left[\chi_c - \chi(t), \gamma_c - \gamma(t) \right]^T $.
 Under this error definition, the aforementioned perturbed system can be transformed into:
\begin{align}
	\label{eq:aircraft_error_model}
	\underbrace{\begin{bmatrix}
			\dot{e}_{\chi}(t) \\ \dot{e}_{\gamma}(t)
	\end{bmatrix}}_{\dot{e}(t)}
	=
	\underbrace{	\begin{bmatrix}
			-\frac{g\tan\phi(t)}{V}\\
			-\frac{g(n_z(t)\cos \phi(t) - \cos(\gamma_c - e_{\gamma}(t)))}{V}
	\end{bmatrix}}_{f_e(e(t),u(t))}
	+ 
	\underbrace{\begin{bmatrix}
			-d_{\chi}\\ - d_{\gamma}
	\end{bmatrix}}_{d_e}
\end{align}
 The Jacobians of $f_e(e,u)$ with respect to $e$ and $u$ are derived as follows:
 \begin{align}
 	\frac{\partial f_e(e,u)}{\partial e}
 	=
 	\begin{bmatrix}
 		0 & 0\\
 		0 & \frac{g}{V}\sin(\gamma_c - e_{\gamma})
 	\end{bmatrix}, 
 	\frac{\partial f_e(e,u)}{\partial u} =\begin{bmatrix}
 		-\frac{g}{V}\sec^2\phi &  0\\
 		\frac{gn_z}{V}\sin\phi & -\frac{g}{V}\cos \phi
 	\end{bmatrix}
 \end{align}
 The relevant model parameters are listed in Table \ref{tab:hyperparameters}.
 \begin{table}[hbtp]
 	\caption{Hyperparameter declarations}
 	\label{tab:hyperparameters}
 	\centering
 	\footnotesize
 	\begin{tabular}{llll|llll}
 		\hline
 		\centering
 		\textbf{Declaration} & \textbf{Param} & \textbf{Value} & \textbf{Unit}
 		&
 		\textbf{Declaration} & \textbf{Param} & \textbf{Value} & \textbf{Unit}
 		\\ 
 		\hline
 		Simulation timespan & $T$ & [0,20] & $\text s$
 		&
 		Lipschitz constant of $d_{\chi}$ & $L_{d_{\chi}}$ & 0.1 & $-$ \\
 		Acceleration of gravity & $g$ & 9.81 & $\text m/\text s^2$ 
 		&
 		Lipschitz constant of $d_{\gamma}$ & $L_{d_{\gamma}}$ & 0.1 & $-$ \\
 		Initial climb angle & $\gamma(0)$ & $\pi/4$& $\text{rad} $
 		&
 		Disturbance frequency of $d_{\chi}$ & $\omega_{\chi}$ & 0.15 & $\text{rad}/ \text s$\\
 		Initial azimuth angle & $\chi(0)$ & $\pi/3$ & $\text{rad} $& Disturbance frequency of $d_{\gamma}$ & $\omega_{\gamma}$ & 0.15 & $\text{rad}/ \text s$\\
 		Initial roll angle & $\phi(0)$ & $\pi/3$ & $\text{rad}$
 		&The range of $\phi$ & [$\phi_{\min}$,$\phi_{\max}$]& [-$\pi/4$,$\pi/4$] & $\text{rad} $\\
 		Initial overload & $n_z(0)$ & $1$& $-$&
 		The range of $\dot{\phi}$ & [$\dot{\phi}_{\min}$,$\dot{\phi}_{\max}$]& [-$\pi/6$,$\pi/6$] & $\text{rad} /\text{s}$\\
 		Reference climb angle & $\gamma_c$ & $\pi/12$ & $\text{rad} $ & 
 		The range of $n_z$ & [$n_{z,\min}$,$n_{z,\max}$]& [-2.1,2.1] & $-$\\
 		Reference azimuth angle & $\chi_c$ & 0& $\text{rad}$&
 		The range of $\dot{n}_z$ & [$\dot n_{z,\min}$,$\dot n_{z,\max}$]& [-1,1] & $/s$\\
 		Reference roll angle & $\phi_c$ & 0& $\text{rad}$&Consolidated velocity & $V$ & 25 & $\text m/ \text s$ 
 		\\
 		Reference overload & $n_{zc}$ & 0& $-$\\
 		\hline 
 	\end{tabular}
 \end{table}
 
 \subsection{Verification of $\gamma$-dissipativity domain for the MIMO-PI Controller}  
 To stabilize the aforementioned perturbed system, we design a MIMO-PI controller as shown in Eq.(\ref{eq:controller}). The optimal controller parameters $K = [K_P^*, K_I^*]$, derived in Ref. \cite{SHENG2025108152}, are given as follows:
 \begin{align}
 	\label{eq:K2}
 		K_P^* = \begin{bmatrix}
 			1.6968  & 0.5906 \\
 			-0.5906  &  1.9556
 		\end{bmatrix}, 
 		K_I^* = \begin{bmatrix}
 			3.4869  & 0.1784 \\
 			-0.1784  &  3.4869
 		\end{bmatrix}
 \end{align}
First, unlike $\mathcal{L}_{K}(\Omega)$ defined on region $\Omega$ containing the origin, we redefine the dissipativity index $\mathcal{L}_{K}(e)$ at an arbitrary point $e$ as:
\begin{align}
	\mathcal{L}_{K}(e) = \lambda_{\max}\left[\mathrm{Re}(A_K(e))\right] + \left\|A_K(e) -A_K(0)\right\|M^2_K(0)
\end{align}
It is evident that $\sup_{e\in \Omega}\mathcal{L}_{K}(e) \leq \mathcal{L}_{K}(\Omega)$, i.e., $\mathcal{L}_{K}(e)$ at any point in region $\Omega$ serves as a lower bound of $\mathcal{L}_{K}(\Omega)$ and directly impacts the specific value of $\mathcal{L}_{K}(\Omega)$. Thus, we can quantify the point-wise dissipativity degree via $\mathcal{L}_{K}(e)$. 

To this end, we select discretized grid points from $-\pi/3$ to $\pi/3$ with a step size of 0.05 along the $e_{\chi}$-axis and from $-\pi/6$ to $\pi/6$ with a step size of 0.01 along the $e_{\gamma}$-axis. $\mathcal{L}_{K}(e)$ values at each point in this rectangular grid region are calculated to evaluate the point-wise dissipativity index, with results shown in Figure \ref{Fig:2}. 
 \begin{figure}[htbp]
	\centering
	\includegraphics[width=0.8\textwidth]{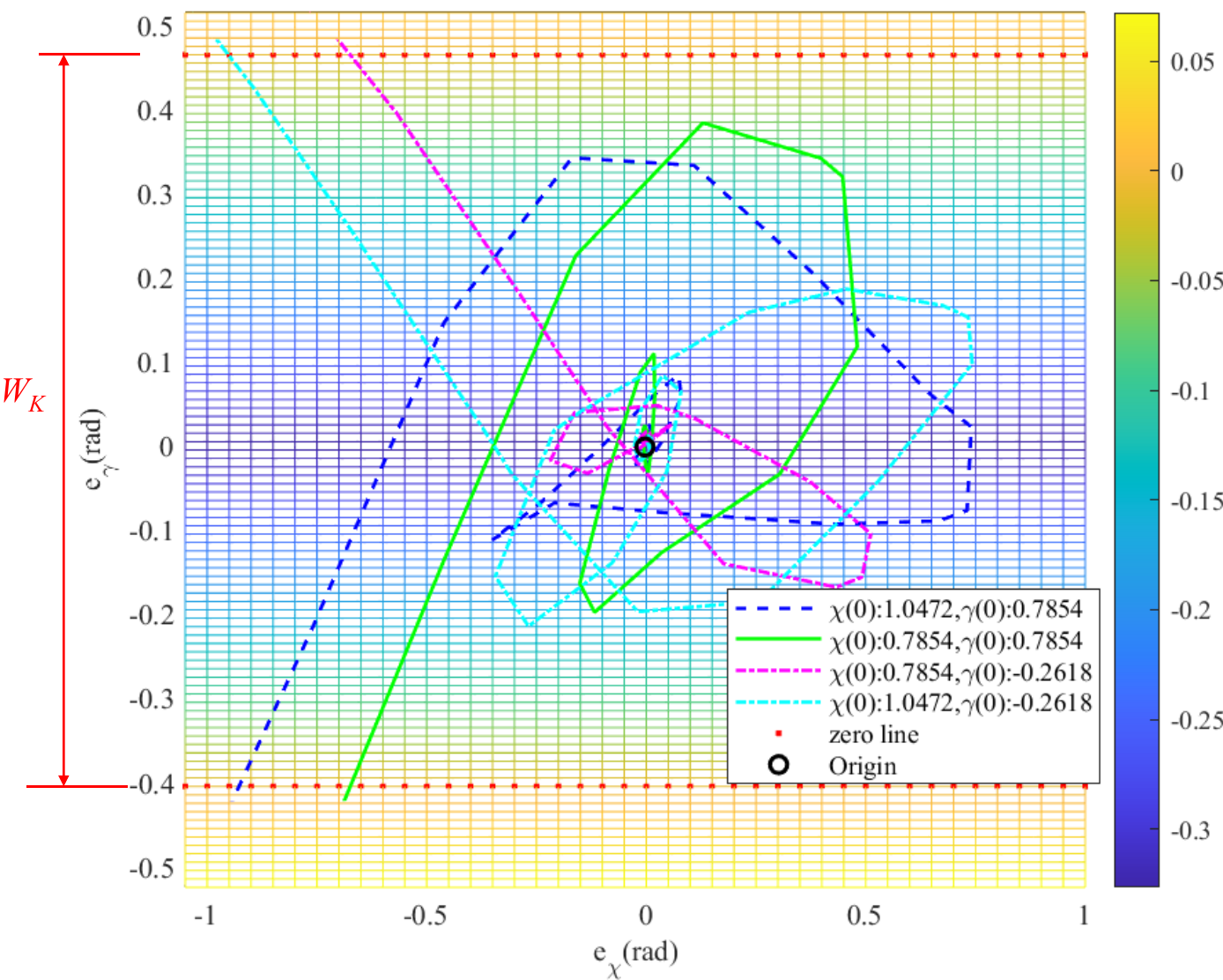}
	\caption{
		Two-dimensional phase portraits of error trajectories under different initial conditions $\chi(0)$ and $\gamma(0)$, which illustrate the error origin, $\mathcal{L}_{K^*}(e)$ at every point $e$, and the corresponding boundary of $\mathcal{L}_{K^*}(\Omega)=0$ (zero line).}
	\label{Fig:2}
\end{figure}
Notably, in the heatmap, $\mathcal{L}_{K}(e)$ exhibits uniformity along the horizontal direction of the $e_{\chi}$-axis, indicating that the error $e_{\chi}$ has no effect on system dissipativity. By contrast, $\mathcal{L}_{K}(e)$ is high at both ends and low in the middle along the horizontal direction of the $e_{\gamma}$-axis, and its distribution is symmetric about the line $e_{\gamma}=0$, which implies that a larger absolute error $|e_{\gamma}|$ degrades system dissipativity and thus undermines robust stabilization. At the boundary where $\mathcal{L}_{K}(e)=0$ (termed the zero line), $\mathcal{L}_{K}(e)>0$ in regions outside this boundary, making it difficult to determine whether such regions possess $\gamma$-dissipativity. 

Leveraging the heatmap characteristics, we use the width $W_K$ between the two zero lines along the $e_{\gamma}$-axis (referred to as the dissipativity domain) to characterize the measure of the dissipative region. In fact, a larger $W_K$ indicates wider regions where the MIMO-PI controller can achieve error $\gamma$-dissipativity under parameter $K$, enabling robust stabilization; conversely, a smaller $W_K$ corresponds to narrower dissipative regions, which is unfavorable for robust stabilization. Further, as shown in Figure \ref{Fig:2}, we design four groups of initial error conditions $e(0)=\left[\chi(0),\gamma(0)\right]^T$. It can be observed that regardless of the initial error, as long as the error trajectory remains within the $\gamma$-dissipativity domain, $e(t)$ will eventually converge to the vicinity of the origin. Moreover, in regions with weak $\gamma$-dissipativity (i.e., smaller $\left\|\mathcal{L}_{K}(e)\right\|$), state trajectories are relatively smooth, implying large fluctuations in the error convergence process; in regions with strong $\gamma$-dissipativity (i.e., larger $\left\|\mathcal{L}_{K}(e)\right\|$), trajectories are tortuous, which means the error convergence process is confined tightly near the origin and thus achieves stable convergence, corresponding to enhanced dissipativity and robustness. 

The above conclusions demonstrate the rationality of using indices $\mathcal{L}_{K}(e)$ and $\mathcal{L}_{K}(\Omega)$ to quantify the $\gamma$-dissipativity degree of the nonlinear perturbed system under the MIMO-PI controller.

 \subsection{Verification of $\gamma$-dissipativity index $\gamma_K(\Omega)$}
 Currently, the widely adopted metric for evaluating the error stabilization performance of the controller is the average Integral Time Absolute Error (ITAE). It is defined as the time integral of the absolute tracking error over a specified time interval from $t=0$ to $T$: 
 \begin{align}
 	\text{ITAE} = \frac{1}{T}\int_{0}^{T}\left\|e(t)\right\| dt
 \end{align}
 We intend to use this metric to demonstrate the robustness quantification capability of the proposed index $\gamma_K(\Omega)$. To validate the rationality of our indictor $\gamma_K(\Omega)$, we derive new controller coefficients \( K \) by introducing incremental perturbations \( \Delta K \) to the baseline optimal coefficients \( K^* \). 
 \begin{align}
 	K = K^* + \Delta K = K^*-\varepsilon \left[I_p,I_i\right]
 \end{align}
 For simplicity, $\Delta K$ can be described in the form $\Delta K = -\varepsilon \left[I_p,I_i\right]$, where $I_p$ and $I_i$ both are identity matrixes, and $\epsilon \in \mathbb{R}^1$ is served as a regulation variable to determine $\Delta K$. We then compare the corresponding response curves and their associated ITAE across multiple $\varepsilon$, as shown in Table \ref{tab:varepsilon}.  
 
 \begin{table}[htbp]
 	\caption{The different rectangular region $\Omega$.}
 	\label{tab:region}
 	\centering
 	\begin{tabular}{llll}
 		\hline
 		Region & Range of $e_{\chi}$ & Range of $e_{\gamma}$ & Color \\
 		\hline
 		$\Omega_1$ & $\left[-0.7,0.7\right]$ & $\left[-0.3,0.3\right]$ & $\text{Magenta}$\\
 		$\Omega_2$ & $\left[-0.5,0.5\right]$ & $\left[-0.22,0.22\right]$ & $\text{Blue}$\\
 		$\Omega_3$ & $\left[-0.25,0.25\right]$ & $\left[-0.15,0.15\right]$ & $\text{Green}$\\
 		$\Omega_4$ & $\left[-0.1,0.1\right]$ & $\left[-0.06,0.06\right]$ & $\text{Red}$\\
 		\hline 
 	\end{tabular}
 \end{table}
 To achieve an objective, gradient-based visualization of the index $\gamma_K(\Omega)$ across different regions $\Omega$, we define four regions $\Omega_1, \Omega_2, \Omega_3, \Omega_4$ enclosed by rectangular boundaries for each parameter $K$, with their boundary lines and colors detailed in Table \ref{tab:region}. Specifically, we calculate the index $\gamma_K(\Omega_i)$ for these four boundary regions, as presented in Table \ref{tab:varepsilon}; their exact ranges are also visualized in Figure \ref{Fig:1}. Notably, both the zero line width $W_K$ and the index $\gamma_K(\Omega_i)$ for each region vary with different controller parameters $K$. A general trend is observed that $\gamma_K(\Omega_i)$ decreases for regions $\Omega_i$ closer to the origin. Since $\gamma_K(\Omega_i)$ theoretically represents the $\mathcal{L}_2$-gain of the error, regions with darker blue in the heatmaps exhibit smaller $\mathcal{L}_2$-gain, indicating stronger robustness against error disturbance energy. However, the parameter $K$ selected by minimizing $\gamma_K(\Omega_i)$ is not necessarily optimal, as a smaller $\gamma_K(\Omega_i)$ may lead to a narrower $W_K$ and thus reduce the scope of $\gamma$-dissipativity. For instance, compared with subfigure (b), subfigure (a) shows a smaller global $\gamma$-dissipativity index but a correspondingly narrower $W_K$.
 
 \begin{table*}[htbp]
 	\caption{The different $\varepsilon$ and corresponding $R_K$, $W_K$, $\gamma_K(\Omega_1)$, $\gamma_K(\Omega_2)$, $\gamma_K(\Omega_3)$, $\gamma_K(\Omega_4)$.}
 	\label{tab:varepsilon}
 	\centering
 	\resizebox{1.0\textwidth}{!}{
 	\begin{tabular}{llllllll|llllllll}
 		\hline
 		$K$ & $\varepsilon$ & $R_K$ & $W_K$ & $\gamma_K(\Omega_1)$ & $\gamma_K(\Omega_2)$ & $\gamma_K(\Omega_3)$ & $\gamma_K(\Omega_4)$
 		 & 
 		$K$ & $\varepsilon$ & $R_K$ & $W_K$ & $\gamma_K(\Omega_1)$ & $\gamma_K(\Omega_2)$ & $\gamma_K(\Omega_3)$ & $\gamma_K(\Omega_4)$\\ 
 		\hline 
 		$K_1$ & -4 & 0.7742 & 0.8 &11.16 & 5.97 &  4.23 & 3.06
 		& 
 		$K_4$ & 0.5 & 0.3626& 0.87 &15.98& 9.31 & 6.82 & 5.07\\
 		$K_2$ & -2 & 0.6754& 0.96 &8.46& 5.72 & 4.42 & 3.41
 		&
 		$K_5$ & 0.8 & 0.2849& 0.69 & 75.92 & 16.92 & 10.05 & 6.68 \\
 		$K_3$ & -1 & 0.5913 & 0.98 & 9.65 & 6.52 & 5.03 & 3.89
 		&
 		$K_6$ & 1 & 0.2192& 0.54 & 30.80& 61.61& 17.03 & 9.04\\    
 		\hline 
 	\end{tabular}}
 \end{table*}
 
 \begin{figure*}[htbp]
 	\centering
 	\begin{tabular}{ccc}
 		\includegraphics[width=0.32\textwidth]{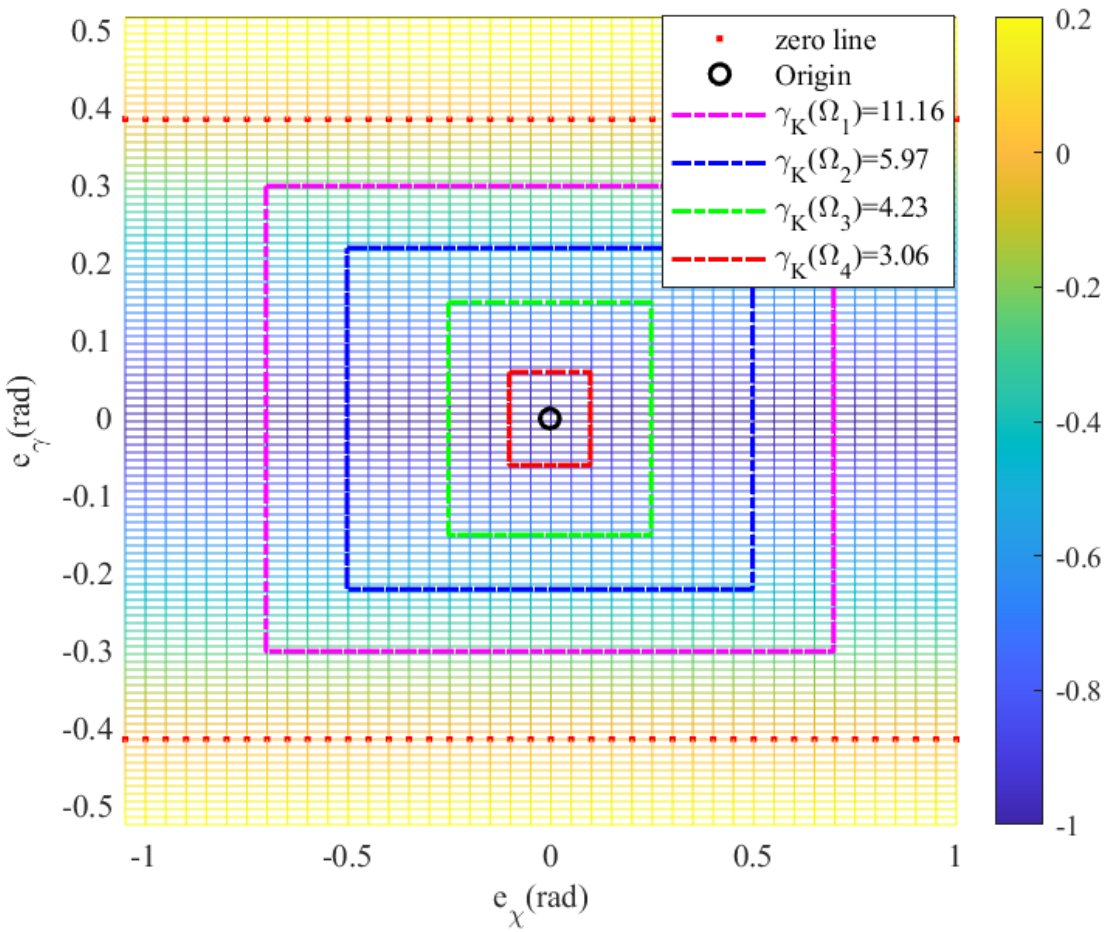} &
 		\includegraphics[width=0.32\textwidth]{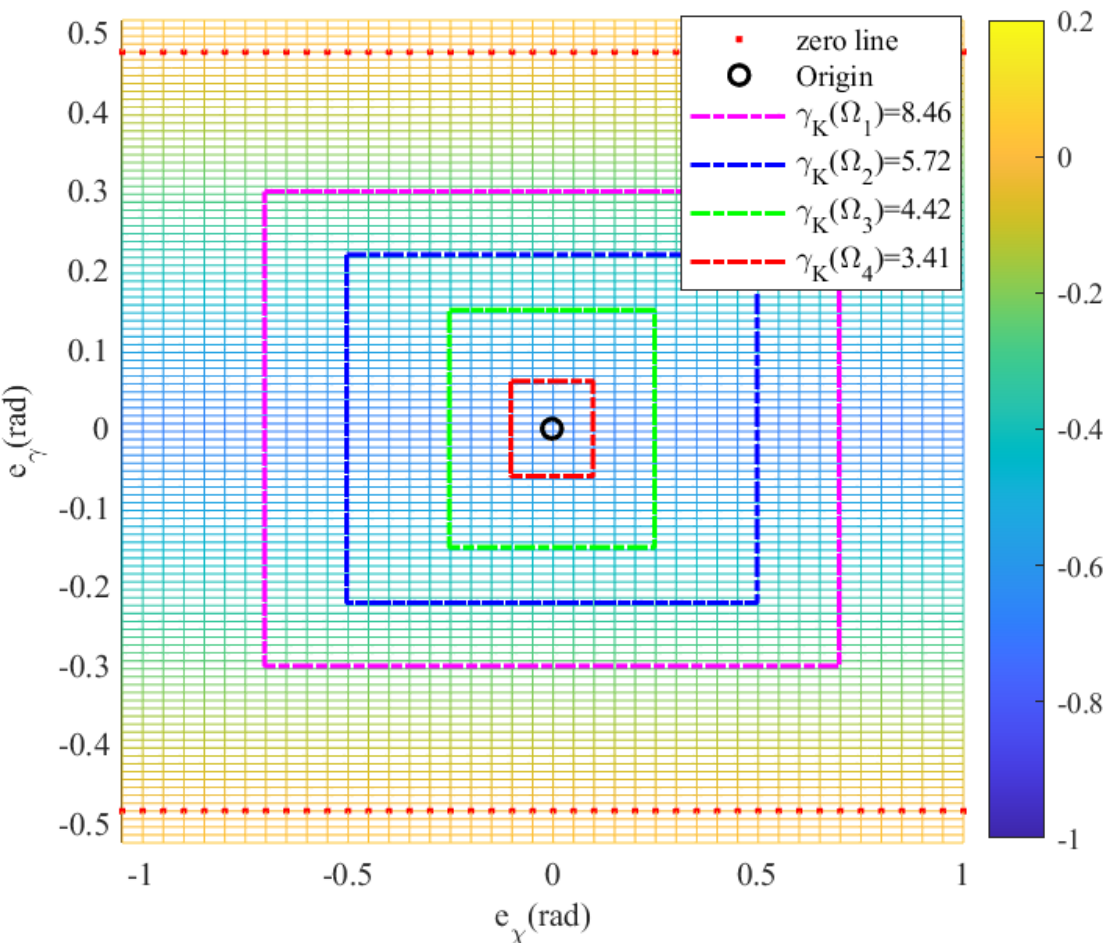} &
 		\includegraphics[width=0.32\textwidth]{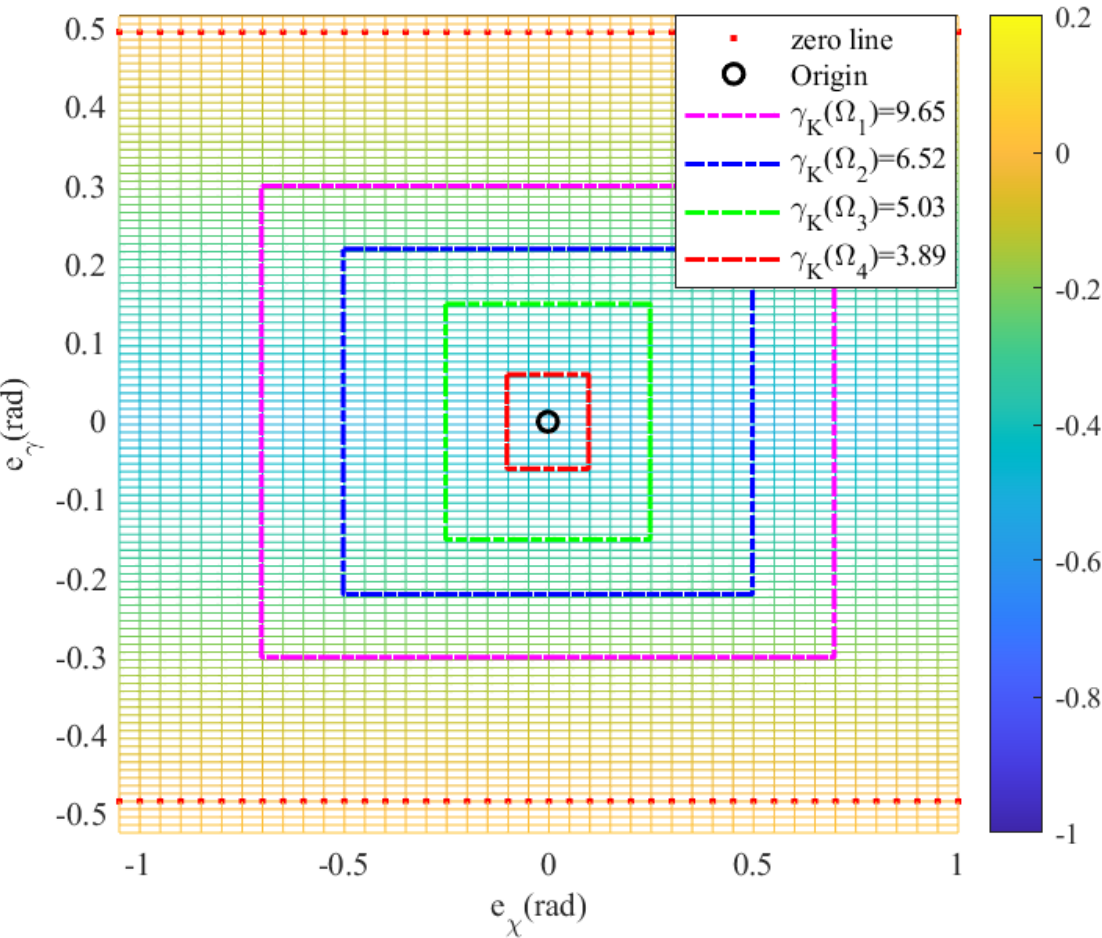}\\
 		(a)  & (b) & (c)  \\
 		\includegraphics[width=0.32\textwidth]{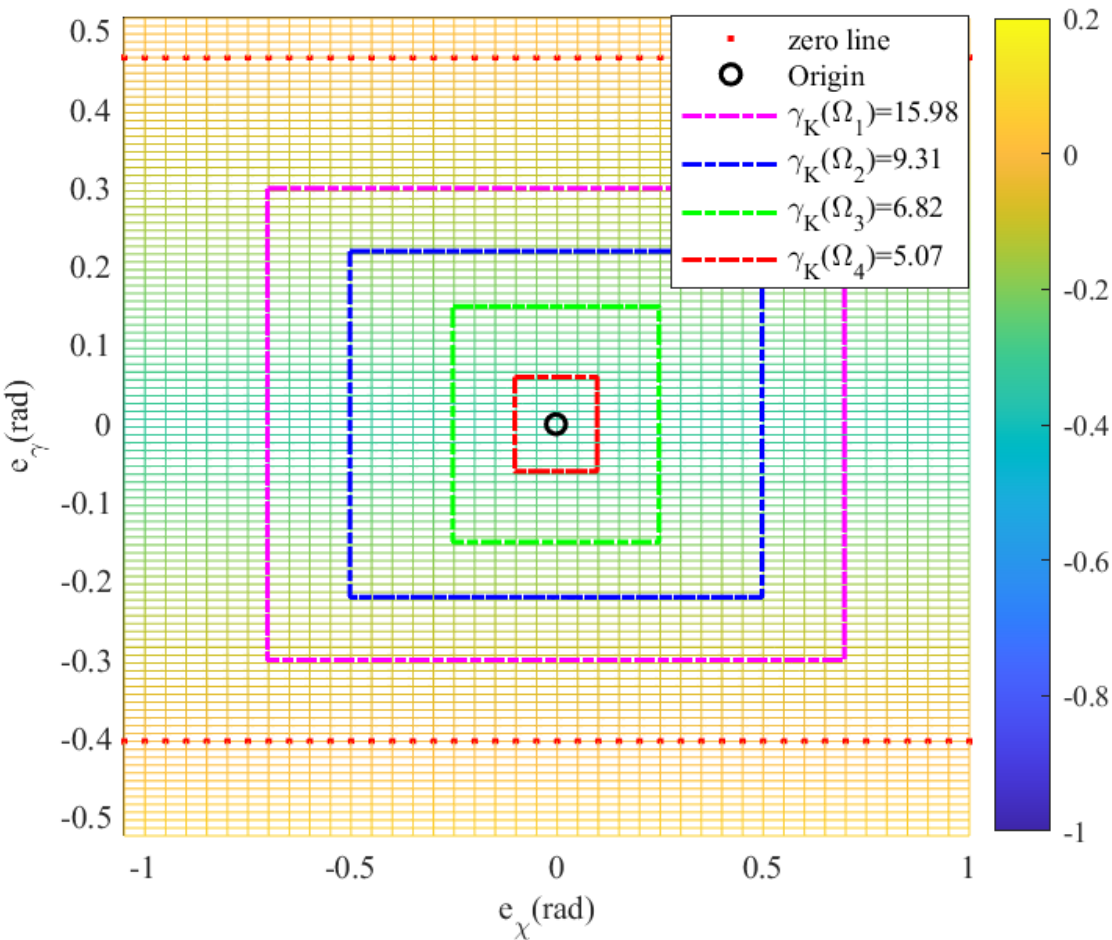} &
 		\includegraphics[width=0.32\textwidth]{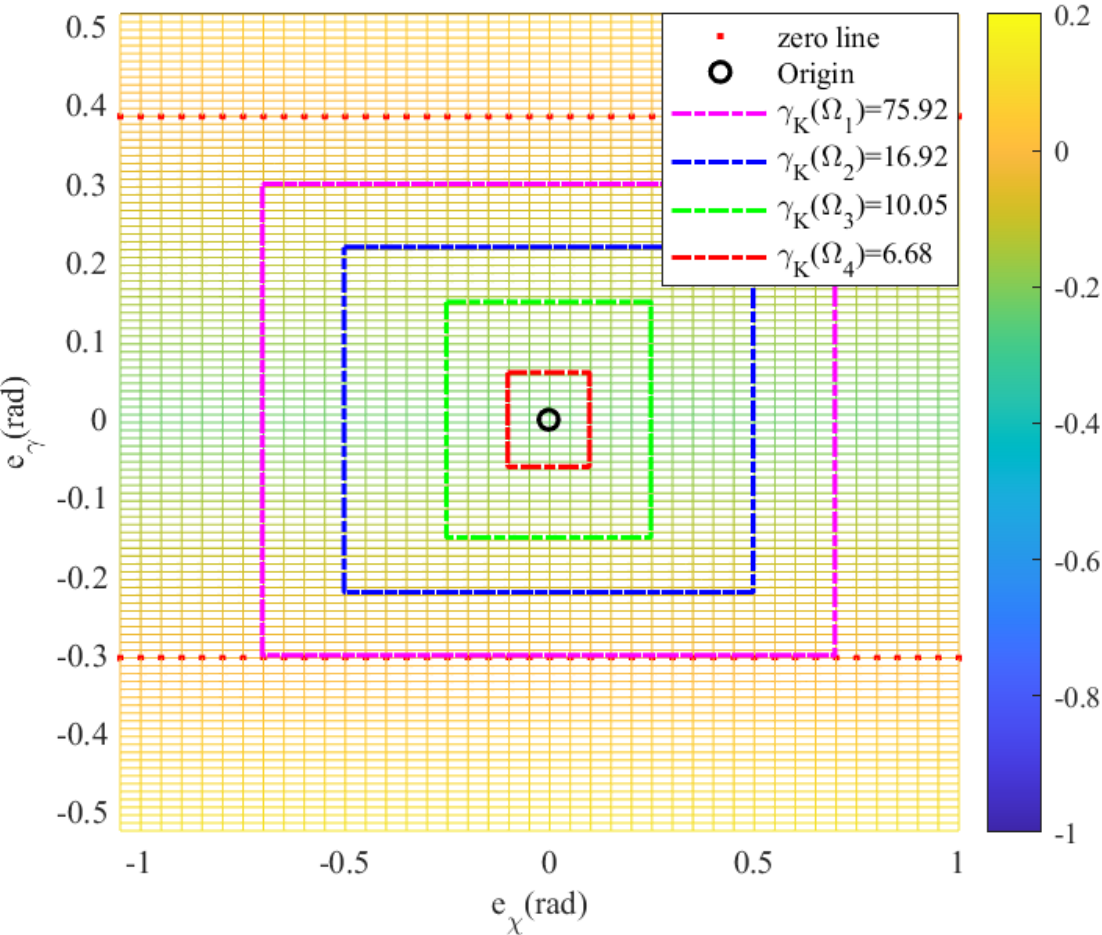} &
 		\includegraphics[width=0.32\textwidth]{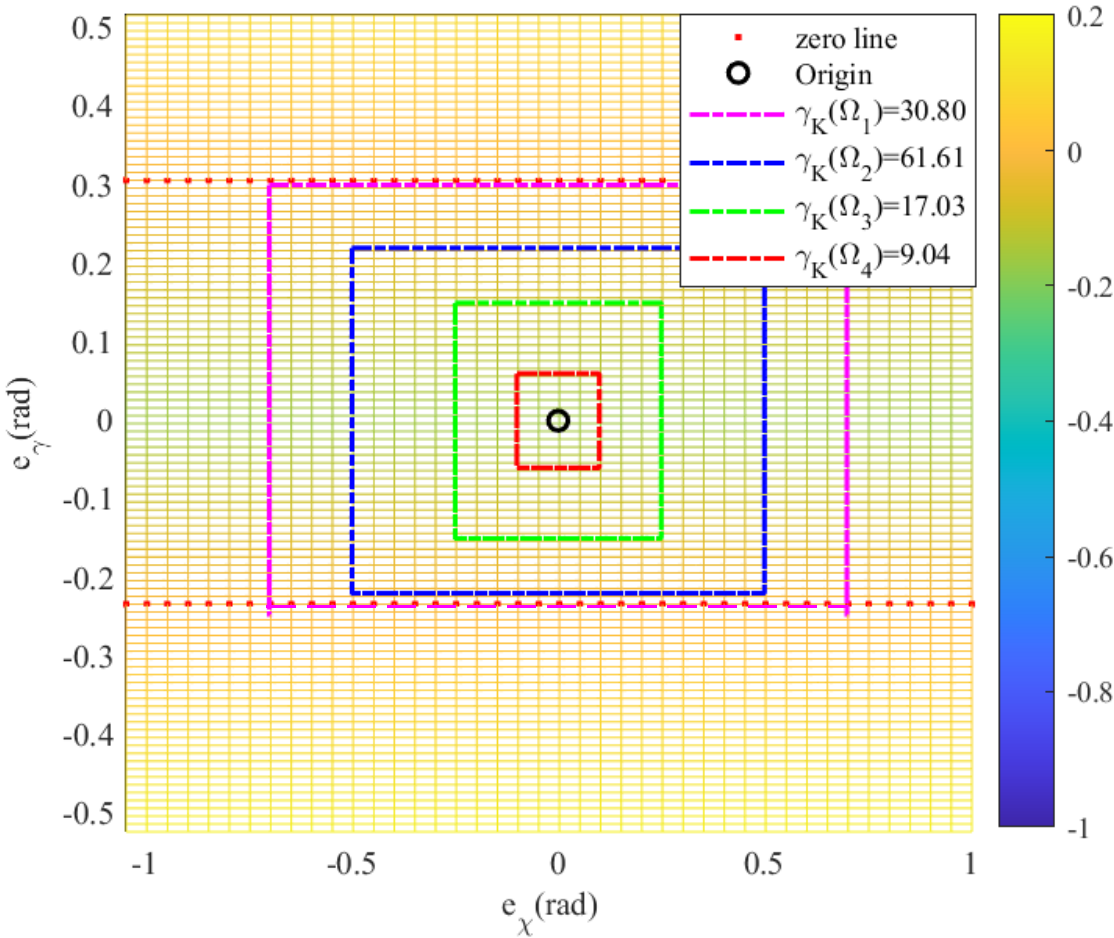}\\
 		(d)  & (e) & (f)  \\
 	\end{tabular}
 	\caption{Dissipativity indices $\gamma_K(\Omega_i)$ over regions $\Omega_i,i=1,2,3,4$ for different controller parameters $K$.
 		(a): For the $K_1$; 
 		(b): For the $K_2$;
 		(c): For the $K_3$;
 		(d): For the $K_4$;
 		(e): For the $K_5$;
 	    (f): For the $K_6$.}
 	\label{Fig:1}
 \end{figure*}
  Figure \ref{Fig:3} provides a more intuitive illustration of the time-domain profiles of the absolute values of the error $e(t)$ and its derivative $\dot{e}(t)$ for different parameters $K$ corresponding to distinct $W_K$ and $\gamma_K(\Omega_i)$. In both subfigures, a smaller $\gamma_K(\Omega_i)$ tends to be associated with weaker oscillations and a smoother error stabilization process. For example, the blue curve, corresponding to the minimum $\gamma_K(\Omega_i)$, shows the flattest profile, whereas the brown-red dashed curve, with the maximum $\gamma_K(\Omega_i)$, exhibits the most severe fluctuations. From an energy perspective, minimizing $\gamma_K(\Omega_i)$ essentially maximizes the attenuation of disturbance energy, thereby reducing oscillation induced by perturbations and enhancing the system's dissipative performance. To further quantify the impact of $W_K$ on error convergence, Figure \ref{Fig:3} also presents the relationships between $W_K$ and the ITAE as well as standard deviation indices of $e(t)$ and $\dot{e}(t)$ for different $K$. A strong negative correlation is clearly observed between $W_K$ and these two indices, implying that a wider $\gamma$-dissipativity domain generally leads to a more stable error stabilization process. Moreover, the standard deviation analysis reveals that $W_K$ has an exponential impact on the stability of the system's convergence process.
  \begin{figure*}[htbp]
 	\centering
 	\begin{tabular}{cc}
 		\includegraphics[width=0.5\textwidth]{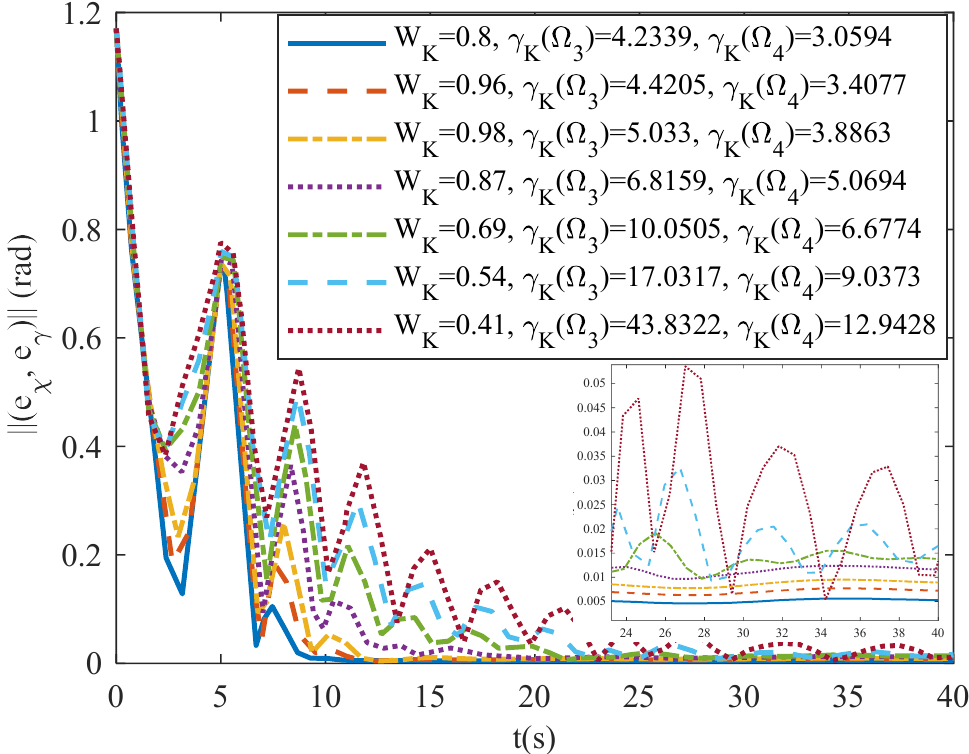} &
 		\includegraphics[width=0.5\textwidth]{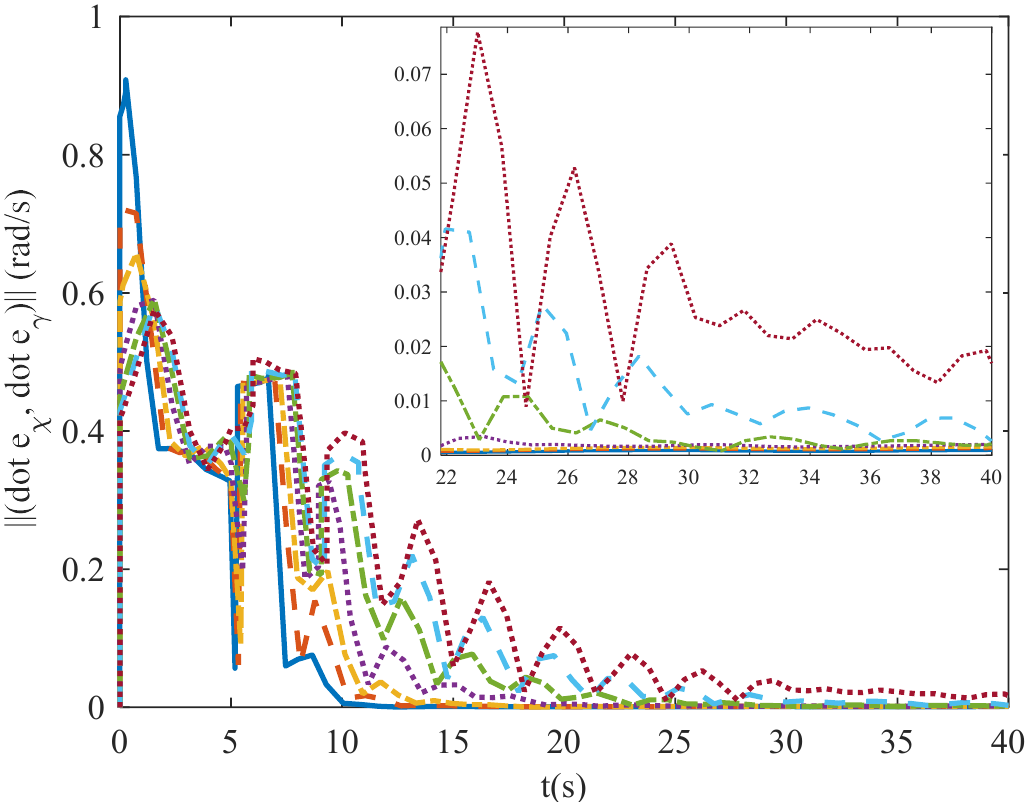}\\
 		(a)  & (b) \\
 	\end{tabular}
 	\caption{ Time curves of the amplitudes of the error $e(t)$ and its derivative $\dot{e}(t)$ at the six different regions in Figure \ref{Fig:1}.
 		(a): Time curves of the amplitudes of the error $e(t)$;
 		(b): Time curves of the amplitudes of the error $\dot{e}(t)$.
 	}
 	\label{Fig:4}
 \end{figure*}

  \begin{figure*}[htbp]
 	\centering
 	\begin{tabular}{cc}
 		\includegraphics[width=0.45\textwidth]{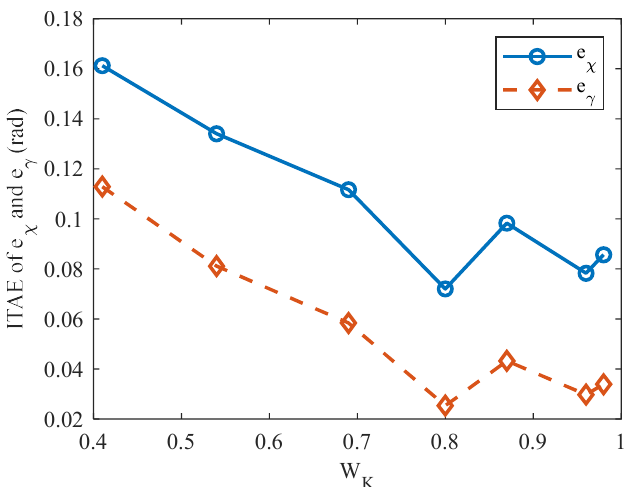} &
 		\includegraphics[width=0.45\textwidth]{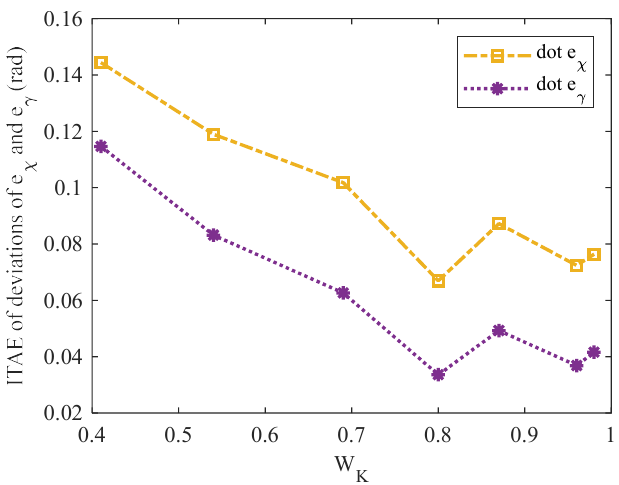}\\
 		(a)  & (b) \\
 		\includegraphics[width=0.45\textwidth]{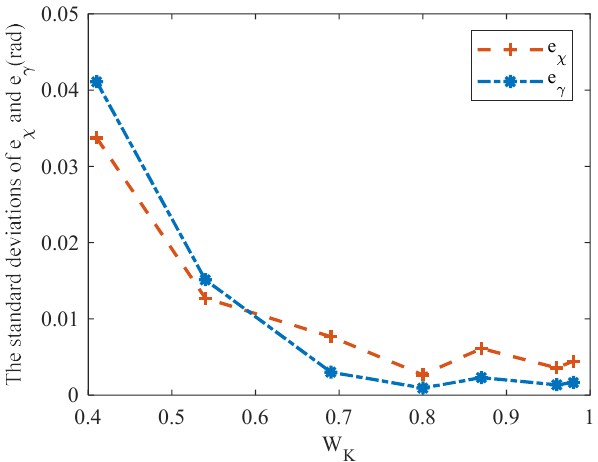} &
 		\includegraphics[width=0.45\textwidth]{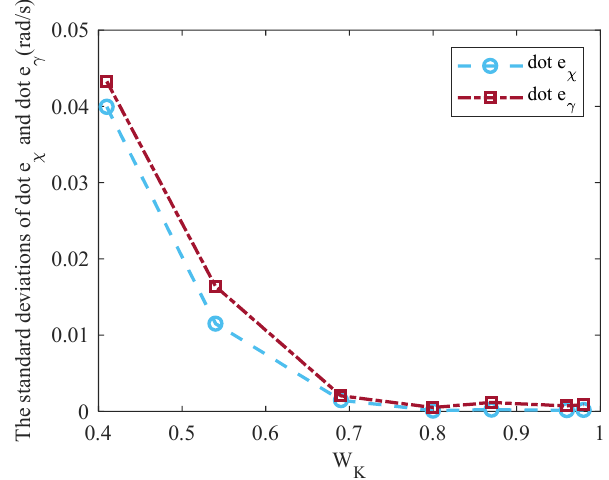}\\
 		(c)  & (d) \\
 	\end{tabular}
 	\caption{ Relationships between indictor $W_K$ (different $K$) and ITAE/Standard deviation of $e_{\chi}, e_{\gamma}, \dot{e}_{\chi}, \dot{e}_{\gamma}$.
 	(a): Relationships between $W_K$ and ITAE of $e_{\chi}$ and $e_{\gamma}$;
 	(b): Relationships between $W_K$ and ITAE of $\dot{e}_{\chi}$ and $\dot{e}_{\gamma}$;
 	(c): Relationships between $W_K$ and standard deviation of $e_{\chi}$ and $e_{\gamma}$;
 	(d): Relationships between $W_K$ and standard deviation of $\dot{e}_{\chi}$ and $\dot{e}_{\gamma}$;
 	}
 	\label{Fig:3}
 \end{figure*}
 
   \begin{figure*}[htbp]
 	\centering
 	\begin{tabular}{cc}
 		\includegraphics[width=0.46\textwidth]{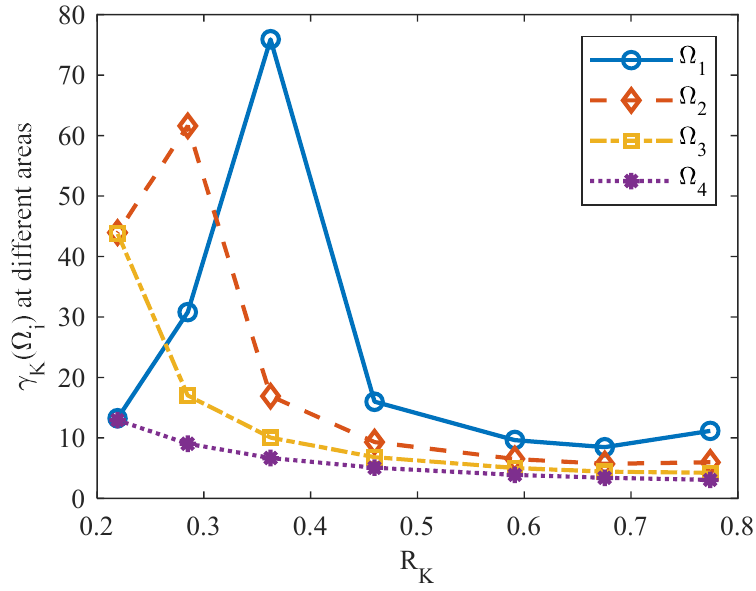} &
 		\includegraphics[width=0.46\textwidth]{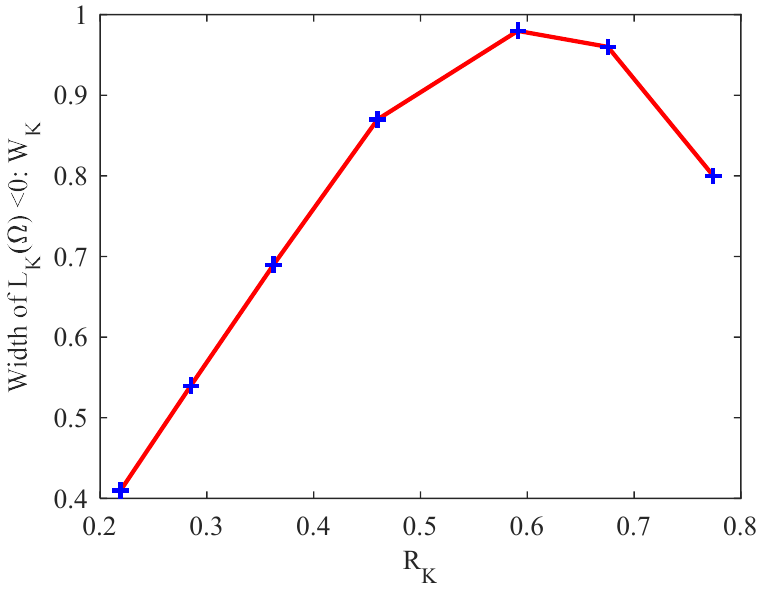}\\
 		(a)  & (b) \\
 	\end{tabular}
 	\caption{ 
 		Curves of the relationships between the robustness index $R_K$ and $\gamma_K(\Omega_i)$ as well as $W_K$ respectively.
 		(a): Relationship between $R_K$ and $\gamma_K(\Omega_i)$, $i=1,2,3,4$;
 		(b): Relationship between $R_K$ and $W_K$.
 	}
 	\label{Fig:5}
 \end{figure*}
  Furthermore, we compare the index $R_K$ proposed in our previous work\cite{SHENG2025108152}—describing the exponential convergence rate of errors under the MIMO-PI controller near the origin—with the indices $\gamma_K(\Omega_i)$ and $W_K$, as shown in Figure \ref{Fig:5}. In subfigure (a), a strong negative correlation between $R_K$ and $\gamma_K(\Omega_i)$ is evident for the origin-adjacent regions $\Omega_3$ and $\Omega_4$, which can be interpreted as stronger $\gamma$-dissipativity near the origin corresponding to a faster exponential convergence rate. However, for the regions $\Omega_1$ and $\Omega_2$ farther from the origin, enhanced dissipativity does not necessarily yield a faster convergence rate, which is mainly attributed to the larger region scope and sustained energy perturbations that may increase system uncertainty. Subfigure (b) further demonstrates the distinctiveness between the proposed index $W_K$ and $R_K$. If we aim to select parameters with both a high exponential convergence rate (larger $R_K$) and a wide $\gamma$-dissipativity domain (larger $W_K$), the peak point in the subfigure undoubtedly corresponds to the optimal parameter $K^*$, as it balances the trade-off between convergence rate and $\gamma$-dissipativity domain. This finding can serve as a key principle for subsequent parameter tuning.

	\section{CONCLUSION}
	To quantify the relationship between the $\mathcal{L}_2$-gain of general perturbed nonlinear systems and corresponding controller parameters, we derive a sufficient criterion for judging whether an origin-containing region satisfies $\gamma$-dissipativity. On this basis, we propose an index to quantify the $\gamma$-dissipativity of the region. With the aid of this index, we can design optimal controller parameters that maximize $\gamma$-dissipativity, i.e., minimize $\mathcal{L}_2$-gain. In the experimental section, we verify the rationality of the proposed index and compare it with our previously developed index characterizing the exponential convergence rate near the origin, which demonstrates their consistency in the origin-adjacent region. However, the proposed index achieves superior quantification performance over a broader range. Notably, in practical deployment, obtaining the Jacobians of the controlled model is essential for optimal parameter design. Future research will focu on Jacobians acquisition under model-free assumptions.

	\Acknowledgements{This work was supported by National Natural Science Foundation of China (Grant No. 51775435).}

	\Supplements{Appendix A.}
	
	\bibliographystyle{unsrt}
	\bibliography{autosam}
	
	\vspace{0.0cm}
	\newpage
	\vspace{-0.1cm}
	\begin{appendix}
		\section{Proof of Lemma \ref{lem:lem1}}
		\label{app:pf_Lemma_1}
		\begin{proof}
			\label{pf:lem1}
			At time \( T \), if the system is \( \gamma \)-dissipativity, then its corresponding positive semi-definite storage function \( V(e) \) for zero-initial state $e(0)=0$ satisfies:  
			\begin{align}
				0\leq V(e(T)) &= V(e(T)) - V(0) \\
				&\leq \int_{0}^{T}\dot{V}(e(\tau)) d\tau \\
				& \leq \int_{0}^{T}\frac{1}{2}\left( \gamma^2 \|\Omega(\tau)\|^2 - \|e(\tau)\|^2 \right) d\tau 
			\end{align}
			Hence, there exists
			\begin{align}
				\sqrt{\frac{\int_{0}^{T}||e(\tau)||^2d\tau  }{\int_{0}^{T}||\Omega(\tau)||^2d\tau } }\leq \gamma
			\end{align}
			The proof is completed.
		\end{proof}
		
		\section{Proof of Lemma \ref{lem:lem3}}
		\label{app:pf_Lemma_3}
		\begin{proof}
			\label{pf:lem3}
			From the \( \gamma \)-dissipativity property of positive semi-definite $V(e(t))$, there exists
			\begin{align}
				\dot{V}(e(t)) \leq \frac{1}{2}\left(\gamma^2||\Omega(t)||^2 - ||e(t)||^2 \right)
			\end{align}
			and
			\begin{align}
				||e(t)||^2\geq \frac{1}{\bar{\lambda}}V(e(t))
			\end{align}
			such that, $\forall t\geq t_0$,
			\begin{align}
				\label{eq:w_t}
				w(t)=\dot{V}(e(t)) + \frac{1}{2\bar{\lambda}} V(e(t)) - \frac{1}{2}\gamma^2||\Omega(t)||^2 \leq 0
			\end{align}
			Hence, solving Eq.(\ref{eq:w_t}) yields:  
			\begin{align}
				\begin{split}
					V(e(t)) &= V(e_0)e^{-\frac{t-t_0}{2\bar{\lambda}}} \\
					&+\int_{t_0}^te^{-\frac{t-\tau}{2\bar{\lambda}}}\left[w(\tau) + \frac{1}{2}\gamma^2||\Omega(\tau)||^2\right]d\tau\\
					&\leq V(e_0)e^{-\frac{t-t_0}{2\bar{\lambda}}} + \frac{\gamma^2}{2}\int_{t_0}^te^{-\frac{t-\tau}{2\bar{\lambda}}} ||\Omega(\tau)||^2d\tau \\
					& = V(e_0)e^{-\frac{t-t_0}{2\bar{\lambda}}} + \frac{\gamma^2}{2}\left(\int_{t_0}^{\frac{t}{2}} + \int_{\frac{t}{2}}^{t}\right)e^{-\frac{t-\tau}{2\bar{\lambda}}} ||\Omega(\tau)||^2d\tau
				\end{split} 
			\end{align}
			There exist $\tau^* \in [t_0,\frac{t}{2}]$ such that
			\begin{align}
				\begin{split}
					0 &\leq \int_{t_0}^{\frac{t}{2}}e^{-\frac{t-\tau}{2\bar{\lambda}}} ||\Omega(\tau)||^2d\tau = e^{-\frac{t-\tau^*}{2\bar{\lambda}}}\int_{t_0}^{\frac{t}{2}} ||\Omega(\tau)||^2d\tau \\
					&\leq C_{\Omega} e^{-\frac{t-\tau^*}{2\bar{\lambda}}}
					\leq C_{\Omega} e^{-\frac{t}{4\bar{\lambda}}}
				\end{split}  
			\end{align}
			Thus, $\lim_{t\rightarrow \infty} \int_{t_0}^{\frac{t}{2}}e^{-\frac{t-\tau}{2\bar{\lambda}}} ||\Omega(\tau)||^2d\tau = 0$. 
			
			In this case, it also holds for $t^*\in[\frac{t}{2},t]$ simultaneously that
			\begin{align}
				\begin{split}
					\int_{\frac{t}{2}}^{t}e^{-\frac{t-\tau}{2\bar{\lambda}}} ||\Omega(\tau)||^2d\tau &=||\Omega(t^*)||^2 \int_{\frac{t}{2}}^{t}e^{-\frac{t-\tau}{2\bar{\lambda}}}d\tau \\	
					&\leq 2\bar{\lambda}||\Omega(t^*)||^2\left(
					1 - e^{-\frac{t}{4\bar{\lambda}}}
					\right)
				\end{split}	 	
			\end{align}
			Due to the fact that $\int_0^{\infty}||\Omega(\tau)||^2 d\tau = C_{\Omega}<\infty$, there exists $||\Omega(t^*)||^2 \rightarrow 0$, as $t^*\rightarrow \infty$. Thus,
			\begin{align}
				\lim_{t\rightarrow \infty}\int_{\frac{t}{2}}^{t}e^{-\frac{t-\tau}{2\bar{\lambda}}} ||\Omega(\tau)||^2d\tau\rightarrow 0
			\end{align}
			Thus,
			\begin{align}
				\lim_{t\rightarrow \infty} V(e(t)) = 0
			\end{align}
			Hence, in the case of $||e(t)||\leq \alpha^{-1}\left( V(e(t)) \right)$, there exists
			\begin{align}
				\lim_{t\rightarrow \infty} ||e(t)|| = 0
			\end{align}
			The proof is completed.
		\end{proof}
		
	\end{appendix}
	
\end{document}